\documentclass{article}
\usepackage{graphicx} 

\usepackage[
    pdftex,
    pdfauthor={Giuseppe De Giacomo, Marco Favorito, Luciana Silo},
    pdftitle={Composition of Nondeterministic and Stochastic Services for LTLf Task Specifications},
    pdfsubject={Artificial Intelligence, Service Composition},
    pdfkeywords={Service Composition, Linear Temporal Logic on Finite Traces, Markov Decision Processes}
]{hyperref}
\usepackage{booktabs} 
\usepackage{xspace}
\usepackage{xcolor}
\usepackage{amssymb}
\usepackage{amsfonts}
\usepackage{amsthm}
\usepackage{macros}
\usepackage{todonotes}
\usepackage{xspace}
\usepackage{mathtools}
\usepackage{upgreek}
\usepackage{wasysym}
\usepackage{tikz}
\usepackage{multirow}
\usepackage{tabularx}
\usepackage{thesisSymbol}
\usetikzlibrary{automata, positioning, arrows, fit}
\tikzset{->,
    >=stealth,
    node distance=7cm,
    every state/.style={thick}
}
\usepackage{comment}
\usepackage{todonotes}
\usepackage{dsfont}
\usepackage[normalem]{ulem}
\usepackage{enumitem}
\usepackage{url}
\usepackage{authblk}
\usepackage{subcaption}

\usepackage[
    commandnameprefix=always,
    markup=default,
    markup=none,
    defaultcolor=blue
]{changes}
\setdeletedmarkup{{\textcolor{red}{\sout{#1}}}}
\usepackage{cleveref}  
\usepackage{biblatex}
\newtheorem{theorem}{Theorem}

\newtheorem{lemma}{Lemma}
\newtheorem{problem}{Problem}

\newtheorem{proposition}{Proposition}

\newcommand{\actions}{\ensuremath{\mathsf{actions}}\xspace}
\renewcommand{\trace}{\ensuremath{\mathsf{trace}}\xspace}
\newcommand{\choices}{\ensuremath{\mathsf{choices}}\xspace}
\newcommand{\states}{\ensuremath{\mathsf{states}}\xspace}
\newcommand{\successful}{\ensuremath{\mathsf{successful}}\xspace}

\newcommand{\orch}{\ensuremath{\gamma}}

\newcommand{\prefixes}{\ensuremath{\mathsf{prefixes}}\xspace}
\newcommand{\histories}{\ensuremath{\mathsf{histories}}\xspace}

\newcommand{\supp}{\ensuremath{\mathsf{Supp}}\xspace}

\newcommand{\win}{\ensuremath{\mathit{Win}}}
\newcommand{\preimg}{\ensuremath{\mathit{PreC}}}

\newcommand{\stoc}{\ensuremath{\tilde}}
\newcommand{\stocC}{\ensuremath{{\stoc\C}}}
\newcommand{\stoch}{\ensuremath{{\stoc{h}}}}

\newcommand{\stocS}{\ensuremath{\stoc\S}}
\newcommand{\cost}{\ensuremath{\mathsf{cost}}\xspace}

\newcommand{\Paths}{\ensuremath{\mathsf{Paths}}}
\newcommand{\prob}{\ensuremath{\mathbb{P}}}
\newcommand{\expected}{\ensuremath{\mathbb{E}}}

\newcommand{\deadsigma}{\ensuremath{\sigma_u}\xspace}
\newcommand{\lastt}{\ensuremath{\mathsf{last}}\xspace}
\DeclareMathSymbol{\dv}{\mathord}{operators}{"3A}
\newcommand{\tk}{\ensuremath{\mathsf{act}}\xspace}

\title{Composition of Nondeterministic and Stochastic Services for \LTLf Task Specifications}
\author[1,2]{Giuseppe De Giacomo}
\author[3]{Marco Favorito}
\author[4]{Luciana Silo}

\affil[1]{University of Oxford, UK}
\affil[2]{Sapienza University of Rome, Italy}
\affil[3]{Banca d'Italia, Italy}
\affil[4]{Camera dei Deputati, Italy}

\date{}

\begin{document}

\maketitle

\begin{abstract}
In this paper, we study the composition of services so as to obtain runs satisfying a task specification in Linear Temporal Logic on finite traces (\LTLf). We study the problem in the case services are nondeterministic and the \LTLf specification can be exactly met, and in the case services are stochastic, where we are interested in maximizing the probability of satisfaction of the \LTLf specification and, simultaneously, minimizing the utilization cost of the services. To do so, we combine techniques from \LTLf synthesis, service composition \emph{à la} Roman Model, reactive synthesis, and bi-objective lexicographic optimization on MDPs. This framework has several interesting applications, including Smart Manufacturing and Digital Twins.
\end{abstract}

\section{Introduction}

The service-oriented computing (SOC) paradigm uses services to support the development of
rapid, low-cost, interoperable, evolvable, and massively distributed applications. Services are considered autonomous, platform-independent entities that can be described, published, discovered, and loosely coupled in novel ways \cite{papazoglou2007service}.
Service composition, i.e. the ability to generate new, more useful services from existing ones, is an active ﬁeld of research in the SOC area and has been actively investigated for over a decade.

Particularly interesting in this context is the so-called Roman Model \cite{berardi2003automatic,DBLP:conf/icsoc/BerardiCGM05,brafman2017service,degiacomo2014automated} where services are \emph{conversational}, i.e., have an internal state and are modelled as finite state machines (FSM), where at each state the service offers a certain set of actions, and each action changes the state of the service in some way.
The designer is interested in generating a new service, called \emph{target}, from the set of existing services specified using an FSM, too. The goal is to see whether the target can be satisfied by properly orchestrating the work of the component service and building a scheduler called the orchestrator that will use actions provided by existing services to implement action requests. 

In this paper, we consider a variant of the Roman Model where the composition is \emph{task-oriented}. Specifically, we are given a task, and we want to synthesize an orchestrator that, on the one hand, reactively chooses actions to form a sequence that satisfies the task and, on the other hand, delegates each action to an available service in such a way that at the end of the sequence, all services are in their final states. We do this in two frameworks. In the first framework, we consider the available services as nondeterministic, in the sense that when the orchestrator delegates to them an action, they will change state in a nondeterministic (devilish vs. angelic) way, as studied, e.g. in \cite{DBLP:conf/icsoc/BerardiCGM05}. In the second framework, we consider services whose response is stochastic, in the sense that delegated actions change the service state according to a probability distribution, as studied, e.g., in \cite{yadav2011decision,brafman2017service}.

Our composition is task-oriented, making it different from the Roman model and more similar to Planning in AI \cite{GeffnerBonet13,McIlraithS02,pistore2005automated}.
However, we draw from the work on declarative process modelling in Business Process Management (BPM) in which the task specification is expressed in Linear Temporal Logic on finite traces (\LTLf) \cite{de2013linear} with the so-called \declare assumption that only one action can be selected at each point in time \cite{pesic2007declare}. 
This gives us a rich way to specify dynamic tasks that extend over time as testified by the \declare paradigm \cite{pesic2007declare} itself, and the recent success of such kind of specification in an advanced form of BPM \cite{ABPMmanifesto23}.

In both the nondeterministic and the stochastic case, we give a formal definition of composition and a provably correct technique to actually solve the composition problem and obtain the orchestrator. Both techniques are  readily implementable.
In the nonderministic framework,  the solution technique is based on finding a winning strategy for a two-player game over a particular \DFA game, as done, for example, in \LTLf synthesis \cite{DBLP:conf/ijcai/GiacomoV15}.
In the stochastic framework, the solution technique relies on solving a bi-objective lexicographic optimization \cite{busatto2023bi} over a special Markov Decision Process \cite{puterman1994markov}, allowing to minimize the services' utilization costs while guaranteeing maximum probability of task satisfaction.
Although this paper has a foundational nature, we observe that these kinds of task-oriented compositions are increasingly becoming important in smart manufacturing \cite{DEGIACOMO2023103916}. 
We briefly discuss this point in the final section of the paper.
%


\section{Preliminaries}
\label{sec:preliminaries}
\textbf{LTL$_f$} is a variant of Linear Temporal
Logic (\LTL) interpreted over finite traces, instead of infinite ones 
\cite{de2013linear}. Given a set $\P$ of atomic propositions, \LTLf formulas $\varphi$ are defined by $\varphi ::= a \mid \lnot \varphi \mid \varphi \wedge \varphi \mid \Next \varphi \mid \varphi \lUntil \varphi$,  where $a$ denotes an atomic proposition in $\P$, $\Next$ is the next
operator, and $\lUntil$ is the until operator. We use abbreviations for
other Boolean connectives, as well as the following: eventually
as $\Diamond \varphi \equiv true \lUntil \varphi$; always as $\square \varphi \equiv \lnot \Diamond \lnot \varphi$;
weak next as $\Wnext \varphi \equiv
\lnot \Next \lnot\varphi$ (note that, on finite traces, $\lnot \Next \varphi$ is not equivalent to
$\Next \lnot \varphi$); and weak until as $\varphi_1 \Wuntil \varphi_2 \equiv (\varphi_1 \Until \varphi_2 \lor
  \Box \varphi_1)$ ($\varphi_1$ holds until $\varphi_2$ or forever).
\LTLf formulas are interpreted on finite traces $a =
a_0 \dots a_{l-1}$ where $a_i$ at instant $i$ is a propositional interpretation over the alphabet $2^\P$, and $l$ is the length of the trace.
An \LTLf formula can be transformed into equivalent \emph{nondeteministic automata} (\NFA) in at most EXPTIME and into an equivalent and \emph{deterministic finite automata} (\DFA) in at most 2EXPTIME \cite{de2013linear}. 
A \DFA is a tuple $\A_\varphi = \langle 2^\P, Q, q_0, F, \delta\rangle$ where: \myi $2^\P$ is the alphabet, \myii $Q$ is a finite set of states, \myiii $q_0$ is the initial state, \myiv $F\subseteq Q$ is the set of accepting states and \myv $\delta:Q\times 2^\P \to Q$ is a total transition function. Note that the \DFA alphabet is the same as the set of traces that satisfies the formula $\varphi$.
An \NFA is defined similarly to \DFA except that $\delta$ is defined as a relation, i.e. $\delta\subseteq Q\times 2^\P\times Q$. 
%
%
In particular, \LTLf is used in declarative process specification in BPM, for example in the system  \textsc{DECLARE} \cite{pesic2007declare}. In this specific case it is assumed that only one proposition (corresponding to an action) is true at every time point. That is: $
\xi_\P = \Box(\bigvee_{a\in\Prop}  a) \land \Box(\bigwedge_{a,b\in\Prop, a
  \neq b} a \limp \lnot b) $.
  We call this the  \emph{\declare assumption}, and we do adopt it in this paper.

\noindent
\textbf{Markov Decision Process.}
A \textit{Markov Decision Process} (MDP) is a tuple $M =
(S, A, P, s_0)$, where: \myi $S$ is a finite set of states, \myii $s_0$ is the initial state, \myiii $A$ is a finite set of actions, and  \myiv$P : (S \times A) \to \Delta(S)$ is the transition probability function.
%
An infinite path $\rho \in (S \times A)^\omega$ is a sequence $\rho = s_0a_1s_1a_2 \dots$, where $s_i \in S$ and $a_{i+1} \in A$ for all $i \in \mathbb{N}$. Similarly, a finite run $\rho \in (S \times A)^* \times S$ is a finite sequence $\rho = s_0a_1s_1a_2 \dots a_{m}s_m$. For any path $\rho$ of length at least $j$ and any $i \leq j$, we let $\rho[i\dv j]$ denote the subsequence $s_ia_{i+1}s_{i+1}a_{i+2} \dots a_{j}s_j$. We use $\Paths_\M^\omega = (S \times A)^\omega$ and $\Paths_\M = (S \times A)^* \times S$ to denote the set of infinite and finite paths, respectively.
A policy $\pi : \Paths_\M \rightarrow A$ maps a finite path $\rho \in \Paths_\M$ to an action $a\in A$. 
%
%
We denote with $\Paths_{\M_\pi}$ the set of finite paths over $\M$ whose actions are compatible with $\pi$.
Given a finite path $\rho = s_0a_1 \dots a_ms_m$, the \emph{cylinder} of $\rho$, denoted by $\Paths_{\M}^\omega(\rho)$, is the set of all infinite paths starting with prefix $\rho$. 
%
%
The \(\sigma\)-algebra associated with MDP $\M$ and a fixed policy $\pi$ is the smallest \(\sigma\)-algebra that contains the cylinder sets \(\Paths_{\M_\pi}^\omega(\rho)\) for all \(\rho \in \Paths_{\M_\pi}\). 
For a state \(s\) in \(S\), a measure is defined for the cylinder sets as $\prob_{\M_\pi,s}(\Paths^\omega_{\M_\pi}(s_0 a_1 s_1 \dots a_m s_m)) = \prod_{k=0}^{m-1} P(s_{k+1}|s_k,a_{k+1})$ if $s_0 = s$ and for all $k$, $a_{k+1} = \pi(\rho[0\dv k])$, otherwise $0$.
We also have \(\prob_{\M_\pi,s}(\Paths^\omega_{\M_\pi,s}(s)) = 1\) and \(\prob_{\M_\pi,s}(\Paths^\omega_{\M_\pi,s}(s'))\!\!=\!\!0\) for \(s'\!\!\neq\!\!s\). 
This can be extended to a unique probability measure \(\prob_{\M_\pi,s}\) on the aforementioned \(\sigma\)-algebra.
In particular, if \(\R \subseteq \Paths_{\M_\pi}\) is a set of finite paths forming pairwise disjoint cylinder sets, then 
$
\prob_{\M_\pi,s}(\bigcup_{\rho \in \R} \Paths^\omega_{\M_\pi}(\rho)) = \sum_{\rho \in \R} \prob_{\M_\pi,s}(\Paths^\omega_{\M_\pi}(\rho))$.
We denote with $\expected_{\M_{\pi},s}(X)$ the expected value of a random variable $X$ with respect to the distribution $\prob_{\M_{\pi},s}$.

\section{Composition of Nondeterministic Services for \LTLf Tasks}
\label{sec:nondet-composition}
We start by presenting our service composition in the case the available services are nondeterministic. Unlike the classical Roman model, we do not have an explicit specification of the target service to realize, but rather, a high-level specification of a task to accomplish expressed as an \LTLf formula. 
We want to accomplish such a task despite the available services having nondeterministic behaviour, as in \cite{DBLP:conf/icsoc/BerardiCGM05}.

\subsection{Nondeterministic Services Framework}
\label{sec:nonstochastic-formalization}
In the Roman Model~\cite{berardi2003automatic}, each \emph{(available) service} is defined as a tuple $\S = \langle \Sigma, A, \sigma_0, F, \delta \rangle$ where:
\myi $\Sigma$ is the finite set of service states,
\myii $A$ is the finite set of services' actions,
\myiii $\sigma_0\in \Sigma$ is the initial state,
\myiv $F\subseteq \Sigma$ is the set of final states  (i.e., states in which the computation may stop but does not necessarily have to),
and \myv $\delta\subseteq \Sigma\times A \times \Sigma$
is the service transition relation.
We use the notation $\sigma \xrightarrow{a} \sigma'$
and $(\sigma, a,\sigma') \in \delta$ interchangeably
when $\delta$ is clear from the context.
For convenience, we define $\delta(\sigma, a) = \{\sigma' \mid (\sigma,a,\sigma')\in \delta\}$,
and we assume that for each state $\sigma\in\Sigma$ and each action $a\in A$, there exist $\sigma'\in \Sigma$ such that $(\sigma,a,\sigma')\in \delta$ (possibly $\sigma'$ is an error state $\deadsigma$ that will never reach a final state). 
Actions in $A$ denote interactions between service and clients.
The behaviour of each available service is described in terms of a finite transition system that uses only actions from $A$.
%


Consider a task specification $\varphi$ expressed in \LTLf over the set of propositions $A$, 
and consider a community of $n$ services $\C = \{\S_1, \dots, \S_n\}$, where each set of actions $A_i\subseteq A$.
%
%
An infinite trace of $\C$ is an infinite alternating sequence of the form $t=(\sigma_{10}\dots\sigma_{n0}),\allowbreak(a_1,o_1),(\sigma_{11}\dots\allowbreak\sigma_{n1}),(a_2,o_2)\dots$, 
where 
for every $0 \le k$, we have 
\myi $\sigma_{ik}\in \Sigma_i$ for all $i\in\{1,\dots,n\}$,
\myii $o_k\in \{1,\dots,n\}$,
\myiii $a_k\in A$, and
\myiv for all $i$, $\sigma_{i,k+1} = \delta_i(\sigma_{ik}, a_{k+1})$ if $o_{k+1}=i$, and $\sigma_{i,k+1} = \sigma_{ik}$ otherwise.

A \emph{history of $\C$} is a finite prefix of a trace of $\C$. 
%
With $|h|=m$, we denote the length of such history,
and with $\lastt(h)$, we denote the service state configuration at the last step: $(\sigma_{1m}\dots\sigma_{nm})$.
%
%
%
%
Given a trace $t$, we call $\states(t)$ \emph{sequence of states} of $t$, i.e. $\states(t) = (\sigma_{10}\dots\sigma_{n0}),(\sigma_{11}\dots\sigma_{n1}),\cdots$.
The \emph{choices} of a trace $t$, denoted with $\choices(t)$, is the sequence of actions in $t$, i.e. $\choices(t) = (a_1,o_1),(a_m,o_m), \dots$.
Note that, due to nondeterminism, there might be many traces of $\C$ associated with the same run.
Moreover, we define the \emph{action run} of a trace $t$, denoted with $\actions(t)$, the projection of $\choices(t)$ only to the components in $A$.
$\states$, $\choices$ and $\actions$ are defined also on history $h$, in a similar way.
 Note that both $\choices(h)$ and $\actions(h)$ are empty if $h=(\sigma_{10}\dots\sigma_{n0})$.

An \emph{orchestrator} is a function $\orch: (\Sigma_1\times\cdots\times\Sigma_n)^* \to A\times \{1\dots n\}$ that, given a sequence of states $(\sigma_{10}\dots\sigma_{n0})\dots(\sigma_{1m}\dots\sigma_{nm})$, returns the action to perform $a \in A$, and the service (actually the service index) that will perform it.
%
%
Next, we define when an orchestrator is a composition that satisfies $\varphi$.
%
%
Given a trace $t$, with $\histories(t)$, we denote the set of prefixes of the 
trace $t$ that ends with a services state configuration.
%
%
%
%
A trace $t$ is an \emph{execution} of an orchestrator $\orch$ over $\C$ if for all $k\ge 0$, we have $(a_{k+1},o_{k+1}) = \orch((\sigma_{10}\dots\sigma_{n0})\dots(\sigma_{1k}\dots\sigma_{nk}))$.
%
%
Let $\T_{\orch,\C}$ be the set of such executions.
Note that due to the nondeterminism of the services, we can have many executions for the same orchestrator, despite the orchestrator being a deterministic function.
If $h\in\histories(t)$ for some (infinite) execution $t\in\T_{\orch,\C}$, we call $h$ a finite execution of $\orch$ over $\C$.
%
%
We say that some finite execution $h$ is \emph{successful}, denoted with $\successful(h)$, if the following two conditions hold: (1) $\actions(h)\models\varphi$, and (2) all service states $\sigma_i \in \lastt(\states(h))$ are such that $\sigma_i\in F_i$.
If for execution $t\in \T_{\orch,\C}$ there exist a finite prefix history $h\in\histories(t)$ such that $\successful(h)$, we say that $t$ is successful.
Finally, we say that an orchestrator $\orch$ \emph{realizes the \LTLf specification $\varphi$ with $\C$} if, for all traces $t\in \T_{\orch,\C}$, $t$ is successful.
%
%
%
Note that the orchestrator, at every step, chooses the action and the (index of the) service to which the action is delegated. In doing so, it guarantees the (complete)  sequence of actions satisfies the \LTLf task specification and that at each step, the action is delegated to a service that can actually carry out the action, despite the nondeterminism of the services and, moreover, when the orchestrator stops, all services are left in their final states. 
Then, the composition problem is:
\begin{problem}[Composition for \LTLf Task Specifications]
    \label{prob:nondet-comp}
    Given the pair $(\C, \varphi)$, where $\varphi$ is an \LTLf task specification over the set of propositions $A$, and $\C$ is a community of $n$ services $\C = \{\S_1,\dots,\S_n\}$, compute, if it exists, an orchestrator $\orch$ that realizes $\varphi$.
\end{problem}

\subsection{Nondeterministic Services Solution Technique}
To synthesize the orchestrator we rely on a game-theoretic technique: i.e.: \myi we build a game arena where the \emph{controller} (roughly speaking the orchestrator) and the \emph{environment} (the service community) play as adversaries; \myii we synthesize a strategy for the controller to win the game whatever the environment does; \myiii from this strategy we will build the actual orchestrator.

Specifically we proceed as follows: 
(1) first, from the \LTLf task specification we compute the equivalent \NFA of an \LTLf formula, $\A_\varphi$; 
(2) in this \NFA we can give the control of the transition to the controller, we do so by considering the \NFA, as a \DFA $\A_\tk$ over an extended alphabet, still under the control of the controller;
(3) compute a \emph{product of such  \DFA $\A_\tk$ with the services, obtaining a new \DFA $\A_{\varphi,\C}$} again extending the alphabet with new symbols, which this time are under the control of the environment;
(4) the \DFA obtained is the arena over which we play is the so-called \emph{\DFA game} \cite{DBLP:conf/ijcai/GiacomoV15};
(5) if a solution of such \DFA game is found, from that solution, we can derive an orchestrator that realizes $\varphi$.
 We now detail each step.

\noindent\textbf{Step 1.}
The \NFA of an \LTLf formula can be computed by exploiting a well-known correspondence between \LTLf formulas and automata on finite words \cite{de2013automatic}. 
In particular, using the \textsc{ltl}$_f$2\textsc{nfa} algorithm \cite{BrafmanGP18}, we can compute an \NFA $\A_\varphi = (A, Q, q_0, F, \delta)$ that is equivalent to the specification $\varphi$ which can be exponentially larger than the size of the formula.
Note that the alphabet of the \NFA is $A$ since we assume the specification satisfies the DECLARE assumption: only one action is executed at each time instant.

\noindent\textbf{Step 2.}
From the \NFA of the formula $\varphi$, $\A_\varphi$, which is on the alphabet $A$, we define a \emph{controllable \DFA} on the alphabet $A\times Q$, $\A_\tk = (A\times Q, Q, a_0, F, \delta_\tk)$, where everything is as in $\A_{\varphi}$ except $\delta_\tk$ that is defined as follows: $\delta_\tk(q, (a, q')) = q'$ iff $(q, a, q')\in \delta$.
Notice that if a sequence of actions is accepted by the \NFA $\A_\varphi$ as witnessed by the run $r = q_0,a_1,\dots,q_n$, then the run itself is accepted by $\A_\tk$.
%
%
Intuitively, with the \DFA $\A_\tk$, we are giving to the controller not only the choice of actions but also the choice of transitions of the original \NFA $\A_\varphi$, so that those chosen transitions lead to the satisfaction of the formula.
In other words, for every sequence of actions $a_1,\dots,a_n$ accepted by the \NFA $\A_\varphi$, i.e. satisfying the formula $\varphi$, there exists a corresponding alternating sequence $q_0,a_1,\dots,q_n$ accepted by the \DFA $\A_\tk$, and viceversa.
This means that when we project out the $Q$-component from the accepted sequences of $\A_\tk$, we get a sequence of actions satisfying $\varphi$.
It can be shown that:

\begin{proposition}
    $a_1\dots a_m\in\L(\A_\varphi)$ iff $(a_1,q_1)\dots(a_m,q_m) \in \L(\A_\tk)$, for some $q_1\dots q_m$.
\end{proposition}
\begin{proof}
    By definition, $a_1\dots a_m\in\L(\A_\varphi)$ iff there exist a run $r=q_1\dots q_m$ s.t. for $1\le k \le m$, $\delta(q_{k-1},a_k) = q_k$ and $q_m \in F$.
    Consider the word $w' = (a_1,q_1)\dots(a_m,q_m)$.
    By construction of $\A_\tk$, $w'$ induces a run $r^d = r$.
    Since $q_m\in F$ by assumption, $r^d$ is an accepting run for $\A_\tk$, and therefore $w' \in \L(\A_\tk)$ is accepted.
    The other direction follows by construction because, if $(a_1,q_1)\dots(a_m,q_m)\in\L(\A_\tk)$, then by construction $q_1\dots q_m$ is a run of $\A_\varphi$ over word $a_1\dots a_m$, and since $q_m\in F$ by assumption $a_1\dots a_m\in \L(\A_\varphi)$.
\end{proof}

\noindent\textbf{Step 3.}
Then, given $\A_\tk$ and $\C$, we build the \emph{composition \DFA} $\A_{\varphi,\C} = (A',Q',q'_0,\allowbreak F', \delta')$ as follows:
%
$A' = \{(a,q,i,\sigma_j) \mid (a,q,i,\sigma_j)\in A \times Q \times \{1,\dots,n\} \times \big(\bigcup_i \Sigma_i\big)$\text{ and }$ \sigma_j\in\Sigma_i$\};
$Q' = Q\times \Sigma_1\times\cdots\Sigma_n$;
$q'_0 = (q_0, \sigma_{10}\dots\sigma_{n0})$;
$F' = F\times F_1\times \cdots \times F_n$;
$\delta'((q,\sigma_1\dots\sigma_i\dots\sigma_n), (a, q', i,\allowbreak\sigma'_i)) = (q', \sigma_1\dots\sigma'_i\dots\sigma_n)$ iff $\delta_i(\sigma_i,a) = \sigma'_i$, and $\delta_\tk(q,(a,q')) = q'$.
Intuitively, the \DFA $\A_{\varphi,\C}$ is a synchronous cartesian product between the \NFA $\A_\varphi$ and the service $\S_i$ chosen by the current symbol $(a,q,i,\sigma)\in A'$.
The ``angelic'' nondeterminism of $\A_\varphi$ and the ``devilish'' nondeterminism coming from the services is cancelled by moving the choice of the next \NFA state and the next system service state in the alphabet $A'$.
%
%
%

It can be shown that there is a relationship between the accepting runs of the \DFA $\A_{\varphi,\C}$ and the set of successful executions of some orchestrator $\orch$ over community $\C$ for the specification $\varphi$.
First, let us define some preliminary notions.
%
Given a word $(a_1,q_1,o_1,\sigma_{o_1})\dots(a_m,q_m,o_m,\sigma_{o_m})\in A'^*$, which induces the run $r = (q_0,\sigma_{10}\dots\sigma_{n0}),\dots,(q_m,\sigma_{1m}\dots\sigma_{nm})$ over $\A_{\varphi,\C}$, we define the history $h = \tau_{\varphi,\C}(w) = (\sigma_{10}\dots\sigma_{n0}),(a_1,o_1),\dots,(\sigma_{1m}\dots\sigma_{nm})$.

\begin{proposition}
\label{thm:w-iff-successful-h}
    Let $h$ be a history over $\C$ and $\varphi$ be a specification.
    Then, $h$ is successful iff there exist a word $w\in A'^*$ such that $h=\tau_{\varphi,\C}(w)$ and $w\in\L(\A_{\varphi,\C})$.
\end{proposition}
\begin{proof}
    We prove both directions of the equivalence separately.
    
    \noindent
    $(\Leftarrow)$ 
    Let $w$ be a word $w=(a_1,q_1,o_1,\sigma_{o_1})\dots(a_m,q_m,o_m,\sigma_{o_m})\in A'^*$
    Assume $w\in\L(\A_{\varphi,\C})$.
    We have that the induced run $r = (q_0,\sigma_{10}\allowbreak\dots\allowbreak\sigma_{n0}),\dots,(q_m,\sigma_{1m}\allowbreak\dots\allowbreak\sigma_{nm})$ over $\A_{\varphi,\C}$ is accepting.
    This means that $q_m\in F$ and $\sigma_{im}\in F_i$ for all $i$.
    Now consider the history $h=\tau_{\varphi,\C}(w) = (\sigma_{10}\dots\sigma_{n0}),(a_1,o_1),\dots(\sigma_{1m}\dots\sigma_{nm})$.
    For every $0\le k < m$, we have by definition of $\delta'$ that $\delta_i(\sigma_{i,k},a_{k+1}) = \sigma_{i,k+1}$ if $i=o_k$, and $\sigma_{i,k} = \sigma_{i,k+1}$ otherwise. Moreover, $o_k\in\{1\dots n\}$, $a_k\in A$ and $\sigma_{ik}\in \Sigma_i$, for all $i\in\{1,\dots,n\}$ and $0\le k < m$. Therefore, $h$ is a valid history.
    Moreover, $h$ is a successful history because $q_m\in F$ iff $a_1,\dots, a_m\models \varphi$, and $\sigma_{im}\in F_i$ for all $i$.

    \noindent
    $(\Rightarrow)$
    Assume that $h=(\sigma_{10}\dots\sigma_{n0}),(a_1,o_1),\dots,(\sigma_{1m}\dots\sigma_{nm})$ is a successful history over $\C$
    Then we both have that \myi $\actions(h) = a_1,\dots,a_m \models \varphi$ and \myii $\sigma_{im}\in F_i$ for all $i$.
    By construction of the \NFA $\A_\varphi$, proposition \myi implies that there exists a run $q_0,\dots,q_m$ where $q_m\in F$.
    Moreover, let $\sigma_1,\dots,\sigma_m$ be the sequence of service states s.t. $\sigma_k = \sigma_{i,k}$ (where $i = o_k$), for $1\le k \le m$.
    Now consider the sequence $w = 
    (a_1,q_1,o_1,\sigma_1),
    \dots,
    (a_m,q_m,o_m,\sigma_m)$.
    %
    By construction, we have $\tau_{\varphi,\C}(w) = h$.
    Moreover, from proposition \myi and proposition \myii, it follows that $w\in\L(\A_{\varphi,\C})$.
    This concludes the proof.
\end{proof}

\noindent
\Cref{thm:w-iff-successful-h} shows that there is a tight relationship between accepting runs of $\A_{\varphi,\C}$ and successful histories over $\C$ for specification $\varphi$.
Now, we consider the \DFA $\A_{\varphi,\C}$ as a \DFA game. First, we will present the definition of \DFA game and how to compute a solution.

\noindent\textbf{Step 4.}
\emph{\DFA games} are games between two players, here called respectively the \emph{environment} and the \emph{controller}, that are specified by a \DFA. 
We have a set of $\X$ of \emph{uncontrollable symbols}, which are under the control of the environment, and a set $\Y$ of \emph{controllable symbols}, which are under the control of the controller. 
A \emph{round} of the game consists of both the controller and the environment choosing the symbols they control. A (complete) \emph{play }is a word in $\X\times \Y$ describing how the controller and environment set their propositions at each round till the game stops.
The \emph{specification} of the game is given by a \DFA $\G$ whose alphabet is $\X\times \Y$. 
%
%
A play is \emph{winning} for the controller if such a play leads from the initial to a final state. 
A \emph{strategy} for the controller is a function $f: \X^* \to \Y$ that, given a history of choices from the environment, decides which symbols $\Y$ to pick next. 
In this context, a history has the form $w=(X_1,Y_1)\cdots (X_m,Y_m)$ .
Let us denote by $w_\X|_k$ the sequence projected only on $\X$ and truncated at the $k$-th element (included), i.e., $w_\X|_k = X_1\cdots X_k$.
A \emph{winning strategy} is a strategy $f: \X^* \to \Y$ such that for all sequences $w=(X_1,Y_1)\cdots (X_n,Y_n)$ with $Y_i = f(w_\X \mid_k)$, we have that $w$ leads to a final state of our \DFA game. 
The \emph{realizability} problem consists of checking whether there exists a \emph{winning strategy}. The \emph{synthesis} problem amounts to actually computing such a strategy.

We now give a sound and complete technique to solve realizability for \DFA games. We start by defining the \emph{controllable preimage} $PreC(\E)$ of a set $\E$ of states $\E$ of $\G$ as the set
of states $s$ such that there exists a choice for symbols $\Y$ such that for all choices of symbols $\X$, game $\G$ progresses to states in $\E$. Formally, $\preimg(\E) = \{s \in S \mid \exists Y \in \Y .\forall X \in \X .\delta(s, (X, Y)) \in \E\}$.
Using such a notion, we define the set $\win(\G)$ of winning
states of a \DFA game $\G$, i.e., the set formed by the states from which the controller can win the DFA game $\G$. Specifically, we define $\win(G)$ as a least-fixpoint, making use of approximates $\win_k(\G)$ denoting all states where the controller wins in at most $k$ steps: (1) $\win_0(\G) = F$ (the final states of $\G$); and (2) $\win_{k+1}(\G) = \win_{k}(\G) \cup \preimg(\win_k(\G))$.
Then, $\win(\G) = \bigcup_k \win_k(\G)$. 
Notice that computing $\win(\G)$ requires linear time in the number of states in $\G$. Indeed, after at most a linear number of steps $\win_{k+1}(\G) = \win_k(\G) = \win(\G)$.
It can be shown that a DFA game $\G$ admits a winning strategy iff $s_0 \in \win(\G)$ \cite{DBLP:conf/ijcai/GiacomoV15}.

%
The resulting strategy is a transducer $T = (\X\times\Y, Q',q'_0, \delta_T,\theta_T)$, defined as follows: $\X\times \Y$ is the input alphabet, $Q'$ is the set of states, $q'_0$ is the initial state, $\delta_T: Q'\times \X\to Q'$ is the transition function such that $\delta_T(q,X) = \delta'(q,(X, \theta(q))$, and $\theta_T: Q\to \Y$ is the output function defined as $\theta_T(q) = Y$ such that if $q\in \win_{i+1}(\G)\setminus\win_{i}(\G)$ then $\forall X.\delta(q,(X,Y))\in \win_i(\G)$ \cite{DBLP:conf/ijcai/GiacomoV15}.


\noindent\textbf{Step 5.}
Given a strategy in the form of a transducer $T$, we can obtain an orchestrator that realizes the specification as follows.
Let the \emph{extended transition function} $\delta_T^*$ of $T$ is $\delta_T^*(q, \epsilon) = q$ and $\delta_T^*(q, wa) = \delta_T(\delta_T^*(q, w), a)$. 
%
Then, for every sequence $w$ of length $m\ge 0$ $w=(X_1,Y_1)\dots(X_m,Y_m)$, where for each index $k$, $Y_k$ and $X_k$ are of the form $(a_k,q_k,o_k)$ and $\sigma_{o_k,k}$ respectively,
we define the orchestrator $\orch_T((\sigma_{10}\dots\sigma_{n0}),(\sigma_{11}\dots\sigma_{o_1,1}\dots\sigma_{n1}),\dots(\sigma_{1m}\dots\sigma_{o_k,m}\dots\sigma_{nm})) = (a_{m+1}, o_{m+1})$, where $(a_{m+1}, q_{m+1}, o_{m+1}) = \theta_T(\delta_T^*(q_0,w))$.



\medskip\noindent
In the following, we show the correctness of our approach, i.e. that we can reduce the problem of service composition for \LTLf task specifications to solving the \DFA game over $\A_{\varphi,\C}$ with uncontrollable symbols $\X = \bigcup_i\Sigma_i$ and controllable symbols $\Y = A\times Q\times \{1,\dots,n\}$.
%

\begin{theorem}
The \emph{service composition realizability} problem with community $\C$ for the satisfaction of an \LTLf task specification $\varphi$ (Problem 1) can be solved by checking whether $q'_0 \in \win(\A_{\varphi,\C})$.
\end{theorem}
\begin{proof}
    The proof is rather technical, therefore we first give an intuitive explanation. 
    Soundness can be proved by induction on the maximum number of steps $i$ for which the controller wins the \DFA game from $q'_0$, building $\orch$ in a backward fashion such that it chooses $(a_k,o_k)\in A'$ that allows forcing the win in the \DFA game (which exists by assumption).
    Completeness can be shown by contradiction, assuming that there exists an orchestrator $\orch$ that realizes $\varphi$ with community $\C$, but that $q'_0\not\in \win(\A_{\varphi,\C})$; the latter implies that we can build an arbitrarily long history that is not successful, by definition of winning region, contradicting that $\orch$ realizes $\varphi$.

    We now provide the full proof of the claim by separately proving the soundness and completeness of our approach.

    \smallskip\noindent
    \textbf{Soundness.} Assume $q'_0 = (q_0,\sigma_{10}\dots\sigma_{n0})\in \win_m(\A_{\varphi,\C})$, i.e. the controller can win in at most $m$ steps; we aim to show that there exists an orchestrator $\orch$ that realizes $\varphi$, i.e.
    for all executions $t\in \T_{\orch,\C}$, at least one prefix $h$ of $t$ is a successful (finite) execution.

    We prove it by induction on the maximum number of steps $i$ for which the controller wins the \DFA game from $q'_0$, building $\orch$ in a backward fashion.

    \noindent
    \emph{Base case} ($k=0$): assume $q'_0\in F'$, i.e. $q'_0$ is already a goal state. Then, the problem is trivially realizable since the task specification is already satisfied (the empty execution is a successful history $h = \epsilon$ for any $\orch$, and in particular, $h$ is the prefix of every infinite execution $t\in\T_{\orch,\C}$).

    \noindent
    \emph{Inductive case}: assume the claim holds for every state $q'_k\in Q'$ that can reach an accepting state in at most $k$ steps, i.e. $q'_k=(q_k,\sigma_{1k}\dots\sigma_{nk})\in \win_k(\A_{\varphi,\C})$, and let $\orch_k$ be the orchestrator that realizes the task specification starting from such states.
    Let $\Delta \win_{k+1}(\A_{\varphi,\C}) = \win_{k+1}(\A_{\varphi,\C}) \setminus \win_k(\A_{\varphi,\C})$. Consider a state $q'_{k+1} = (q_{k+1},\sigma_{1,k+1}\dots\sigma_{n,k+1})\in \Delta \win_{k+1}(\A_{\varphi,\C})$. 
    By construction, $q'_{k+1}\in \preimg(\win_k(\A_{\varphi,\C}))$. This means that there exist a controllable symbol $Y\in \Y$ such that for all uncontrollable symbols $X\in\X$ we can reach a state in $\win_k(\A_{\varphi,\C})$, i.e. $\delta(q'_{k+1}, (X,Y)) = q'_k \in \win_k(\A_{\varphi,\C})$.
    Let $\orch_{k+1}$ be defined as $\orch_k$ for histories of length $l \le k$, plus $\orch_{k+1}((\sigma'_{1,k+1}\dots\sigma'_{n,k+1})) = (a_{k+1},o_{k+1})$ where $Y=(a_{k+1}, q'_{k+1}, o_{k+1})$, for every $\sigma'_{1,k+1}\dots\sigma'_{n,k+1}$ such that $(q'_{k+1},\sigma'_{1,k+1}\dots\sigma'_{n,k+1}) \in \Delta \win_{k+1}(\A_{\varphi,\C})$.
    By the inductive hypothesis, we have that all finite executions $h_{k}$ of $\orch_k$, starting from $\sigma_{1k}\dots\sigma_{nk}$, are successful finite executions. 
    To prove our claim, we only need to show that all $h_{k+1}$ is also a successful execution, and therefore a prefix of some infinite execution $t\in\T_{\orch,\C}$.
    If $|h_{k+1}| \le k$, then it holds by inductive hypothesis.
    For histories of length $k+1$, this is the case because, by construction, we have $h_{k+1} = \sigma'_{z,k+1},(a_{k+1},o_{k+1}),h' \in \H_{\orch_{k+1},\C}$, for some $h'\in\H_{\orch_{k},\C}$, some $(q'_{k+1}, \sigma'_{z,k+1})\in \Delta \win_{k+1}(\A_{\varphi,\C})$, and $(a_{k+1},o_{k+1}) = \orch_{k+1}(\sigma'_{z,k+1})$. In other words, $h_{k+1}$ is a valid finite execution of $\orch_{k+1}$, and moreover, it is successful since $h'$ is successful.
    Finally, we have that $\orch_{m}$, by induction, is an orchestrator that can force the win of the game from $q'_0$.

    \smallskip\noindent
    \textbf{Completeness.}
    By contradiction, assume there exists an orchestrator $\orch$ that realizes $\varphi$ with community $\C$, but that $q'_0\not\in \win(\A_{\varphi,\C})$.
    If $q'_0\not\in \win(\A_{\varphi,\C})$, then it means, by definition of $\win(\A_{\varphi,\C})$ and $\preimg$, that for all $Y\in\Y$, there exist $X\in\X$ such that the successor state $q'_1$ is not in $\win(\A_{\varphi,\C})$. Therefore, we can recursively generate an arbitrarily long word $w = (X_1,Y_1)\dots(X_m,Y_m)$, for any choice of $Y_1\dots Y_m$, such that $w\not\in \L(\A_{\varphi,\C})$.
    Now consider the history $h=(\sigma_{10}\dots\sigma_{n0}),(a_1,o_1),(\sigma_{11}\dots\allowbreak\sigma_{n1})\dots(a_m,o_m),(\sigma_{1m}\dots\sigma_{nm})$, built as follows:
    for all $1\le k\le m$, $(a_k,o_k) = \orch((\sigma_{10}\dots\sigma_{n0})\dots(\sigma_{1,k-1}\dots\sigma_{n,k-1}))$,
    and for any $Y_k=(a_k,q_k,o_k)$, take $X_k = \sigma$ such that $w = (Y_1,X_1)\dots(Y_k,X_k)\not\in \L(\A_{\varphi,\C})$ and $\delta_{o_k}(\sigma_{o_k,k}, a_k) = \sigma$,
    and set $\sigma_{o_k,k} = \sigma$.
    In other words, we consider any valid execution $h$ of $\orch$ where the successor state of the chosen service is taken according to the winning environment strategy (which is losing for the agent).
    By construction, $h$ is a finite execution of
    $\orch$. 
    Moreover, 
    $h$ it is not successful, because $w$ is not accepted by $\A_{\varphi,\C}$ and, by construction of $\A_{\varphi,\C}$, this means that either $\actions(h)\not\models \varphi$, or for some service $\S_i$, $\sigma_{ik}\not\in F_i$.
    Since we assumed that $\orch$ realizes $\varphi$ with community $\C$, we can construct an infinite execution $t\in\T_{\orch,\C}$ of $\orch$, such that $h$ is a prefix of $t$, that is not successful, we reached a contradiction.
\end{proof}

\smallskip

\noindent
\textbf{Computational cost.} Considering the cost of each of the steps above, we get the following upper bound for the worst-case computational cost.
\begin{theorem}
    Problem \ref{prob:nondet-comp} can be solved
    in at most exponential time in the size of the formula, 
    in at most exponential time in the number of services, 
    and in polynomial time in the size of the services.
\end{theorem}
\noindent
We conjecture this cost is actually the  exact complexity characterization of the problem, but proper proof is left as future work.
%
%
Moreover, note that this is analogous to what happens in the classic service composition setting, but with the difference that one must realize all traces requested by the target service, whereas in our case only one trace is enough.

%
%
%
%
%
%
%
%
%
\section{Composition of Stochastic Services for \LTLf Tasks}

\label{sec:stochastic-composition}
In this section, we present our service composition framework in stochastic settings. 
Similarly to \Cref{sec:nondet-composition}, we aim to realize an \LTLf task specification with the available services.
However, this time we model nondeterminism using probability distributions over the services' successor states. Moreover, we allow the specification to be approximately satisfied by considering the objective of maximizing the satisfaction probability of the specification.

\subsection{Stochastic Services Framework}
In this setting, we still have a task specification $\varphi$ expressed in \LTLf over the set of propositions $A$, but we consider \emph{stochastic services}, i.e. services whose result of the operation has a probabilistic outcome.
A stochastic service is a tuple $\stocS = \langle \Sigma, A, \sigma_{0}, F, P, C\rangle$, where $\Sigma$, $A$, $\sigma_0$, and $F$ are defined as in the non-stochastic setting, $C: \Sigma \times A \to \mathbb{R}^+$ is the \emph{cost function} that assigns a (strictly positive) cost to each state-action pair, and $P: \Sigma \times A \to \Delta(\Sigma)$ is the \emph{transition function} that returns for every state $\sigma$ and action $a$ a distribution over next states.
%
%
A \emph{(stochastic) service community} is a collection of stochastic services $\stocC = \{ \stocS_1,\dots,\stocS_n \}$.
A \emph{(stochastic) trace of $\stocC$} is an infinite alternating sequence of the form $\stoc t = (\sigma_{10}\dots\sigma_{n0}),(a_1,o_1),(\sigma_{11},\dots,\sigma_{n1}),\dots$, where $\sigma_{i0}$ is the initial state of every service $\S_i$ and, for every $k \ge 1$, we have \myi $\sigma_{ik} \in \Sigma_i$ for all $i \in \{1,\dots,n\}$, \myii $o_k \in \{1,\dots, n\}$, \myiii $a_k \in A$, and \myiv for all $i$, $\sigma_{ik}\in \supp(P_i(\sigma_{i,k-1},a_{ik}))$ if $o_k = i$, and $\sigma_{ik} = \sigma_{i,k-1}$ otherwise.
The definitions of histories, sequence
of states, choices, orchestrators, executions, successful executions, and realizability remain the same.


We are interested in orchestrators that maximize the probability of satisfaction of the task specification, even when the specification cannot be surely satisfied (e.g. due to a stochastic misbehaviour of some service).
Moreover, while guaranteeing the optimal probability of satisfaction, we aim to find those orchestrators that minimize the expected utilization cost of the services, conditioned on the achievement of the task.

Before proceeding with a formalization of the optimization problem, we introduce additional auxiliary notions.
%
%
%
Analogously to what has been done for MDPs, for a finite execution $h$ of $\orch$ over $\stocC$, we use $\T_{\orch,\stocC}(h)$ to denote the set of all (infinite) executions $t\in\T_{\orch,\stocC}$ such that $h\in\histories(t)$.
Moreover, the $\sigma$-algebra associated with the stochastic behaviour of the orchestrator $\orch$ over the stochastic community $\stocC$ is the smallest $\sigma$-algebra that contains the trace sets $\T_{\orch,\stocC}(h)$, for all finite executions $h$, with the unique probability measure over it defined as:
\begin{equation}
\label{eq:probability-of-history}
\prob_{\orch,\stocC}(h) = \prod\limits_{k=1}^{|h|} P_{o_k}(\sigma_{o_k,k} \mid \sigma_{o_k,k-1}, a_{k})
\end{equation}
%
%
In particular, note that $\prob_{\orch,\stocC}(\T_{\orch,\stocC}(\langle(\sigma_{10}\dots\sigma_{n0})\rangle)) = 1$.
Let $\H^\varphi_{\orch,\stocC}$ be the set of finite executions $h$ of $\orch$ on $\stocC$ that start from $\sigma_{10}\dots\sigma_{n0}$ such that \myi they are successful, and \myii there is no prefix history $h$ that is successful. 
%
%
Intuitively, such a set only contains the executions that are successful for the first time.
%
%
%
The \emph{satisfaction probability} of $\varphi$ under orchestrator $\orch$ and community $\stocC$ is given by:
\begin{equation}
\label{eq:prob-sat-goal}
    \P^\stocC_\varphi(\orch) = \prob_{\orch,\stocC}\bigg(\bigcup_{h\in\H^\varphi_{\orch,\stocC}} \T_{\orch,\stocC}(h) \bigg)
\end{equation}
It is crucial to observe that since by definition there is no pair $h',h''\in \H^\varphi_{\orch,\stocC}(h)$ such that $h'\in\prefixes(h'')$, all trace sets $\T_{\orch,\stocC}(h)$ for $h\in \H^\varphi_{\orch,\stocC}$ are pairwise disjoint sets, which means that $\P^\stocC_\orch$ is a well-defined probability.

Moreover, we define the \emph{(conditioned) expected utilization cost} of services as the expected cost an orchestrator incurs in its successful executions, i.e.:
\begin{equation}
\label{eq:min-exp-cost}
    \J^\stocC_\varphi(\orch) = \mathbb{E}_{h\sim \prob_{\orch,\stocC}}\Bigg[ \sum\limits_{k=1}^{|h|} C_{o_k}(\sigma_{o_k,k-1},a_{k}) \bigg\vert\  \successful(h) \Bigg]
\end{equation}
Let $\Gamma(\stocC)$ be the set of orchestrators for the community $\stocC$.
Let $f: \Gamma(\stocC) \to \mathbb{R}$ be an objective function. We say an orchestrator $\orch\in\Gamma(\stocC)$ is $f$-optimal if $f(\sigma) = \sup_{\tau\in \Gamma(\stocC)} f(\tau)$, and write $\Gamma_f$ for the set of $f$-optimal orchestrators.

Finally, we define our optimization problem. We want to compute an orchestrator $\orch$ such that the following holds: 
\begin{equation}
\label{eq:lexicographic-objective}
    \orch\in\Gamma_{\P^\stocC_\varphi} \text{ and } \J^\stocC_\varphi(\orch) = \inf\limits_{\tau\in\Gamma_{\P^\stocC_\varphi}} \J^\stocC_\varphi(\tau)
\end{equation}
Intuitively, we fix a lexicographic order on the objective functions $\P^\stocC_\varphi$ and $\J^\stocC_\varphi$, meaning that we aim to minimize the expected utilization cost to satisfy the specification, conditioned to the satisfaction of the specification,
while guaranteeing the optimal probability of satisfying it.
Interestingly, in case the specification is \emph{exactly} realizable (in the sense of \Cref{sec:nondet-composition}), the notion of optimal orchestrator according to \Cref{eq:lexicographic-objective} coincides with the notion of realizability, as shown in the following results.

\begin{lemma}
    \label{lemma:union-sets-is-full-cylinder}
    If $\orch$ realizes the specification $\varphi$ over $\stocC$, then $\bigcup_{h\in\H^\varphi_{\orch,\stocC}} \T_{\orch,\stocC}(h) = \T_{\orch,\stocC}(\langle(\sigma_{10}\dots\sigma_{n0})\rangle)$.
\end{lemma}
\begin{proof}
    We prove \myi $\bigcup_{h\in\H^\varphi_{\orch,\stocC}} \T_{\orch,\stocC}(h) \subseteq \T_{\orch,\stocC}(\langle(\sigma_{10}\dots\sigma_{n0})\rangle)$
    and
    \myii $\bigcup_{h\in\H^\varphi_{\orch,\stocC}} \allowbreak\T_{\orch,\stocC}(h) \supseteq \T_{\orch,\stocC}(\langle(\sigma_{10}\dots\sigma_{n0})\rangle)$
    separately.
    Proposition \myi is immediate: every execution belongs to the set of executions starting from $\sigma_{10}\dots\sigma_{n0}$.
    To prove proposition \myii, we start by observing that for all $t\in\T_{\orch,\stocC}(\langle(\sigma_{10}\dots\sigma_{n0})\rangle)$, by definition of realizing orchestrator, they have a prefix $h'\in \prefixes(h)$ that is successful. 
    In particular, if $h''$ is the shortest prefix of $h$ that is successful, then $h''\in\H^\varphi_{\orch,\stocC}$ and $t\in\T_{\orch,\stocC}(h'')$.
    This implies that $t\in\bigcup_{h''\in\H^\varphi_{\orch,\stocC}} \T_{\orch,\stocC}(h'')$.
\end{proof}

\begin{theorem}
    \label{thm:optimality-iff-realizability}
    Let $\stocC$ be a community of stochastic services, and $\varphi$ be a task specification.
    %
    %
    The orchestrator $\orch$ realizes $\varphi$ with community $\stocC$ iff $\P^\stocC_\varphi(\orch) = 1$.
\end{theorem}
\begin{proof}
    \noindent
    ($\Rightarrow$)
    If an orchestrator $\orch$ realizes $\varphi$, then all infinite executions $t\in\T_{\orch,\stocC}$ have a prefix $h'\in\histories(t)$ that is successful. Let $h''$ the shortest of such prefixes.
    This implies that $t\in \bigcup_{h''\in\H^\varphi_{\orch,\stocC}} \T_{\orch,\stocC}(h'')$.
    By \Cref{lemma:union-sets-is-full-cylinder}, this set is equal to $\T_{\orch,\stocC}(\langle(\sigma_{10}\dots\allowbreak\sigma_{n0})\rangle)$.
    Since by definition $\supp(\prob_{\orch,\stocC}) \subseteq \T_{\orch,\stocC}(\langle(\sigma_{10}\dots\sigma_{n0})\rangle)$,
    we have that $\P^\stocC_\varphi(\orch) = \prob_{\orch,\stocC}(\bigcup_{h\in\H^\varphi_{\orch,\stocC}} \T_{\orch,\stocC}(h)) = \prob_{\orch,\stocC}(\T_{\orch,\stocC}(\langle(\sigma_{10}\dots\allowbreak\sigma_{n0})\rangle)) = 1$.

    \noindent
    ($\Leftarrow$) Assume an orchestrator $\orch$ is such that $\P^\stocC_\varphi(\orch) = 1$.
    This implies that for all orchestrator infinite executions $t\in\supp(\prob_{\orch,\stocC}) \subseteq \T_{\orch,\stocC}(\langle(\sigma_{10}\dots\sigma_{n0})\rangle)$, there is a prefix $h\in\histories(t)$ such that $h\in \H_{\orch,\stocC}^\varphi$ and $t\in \T_{\orch,\stocC}(h)$. 
    This means $t$ is successful, and therefore, all executions are successful, i.e. the definition of realizability.
\end{proof}

\begin{theorem}
    \label{thm:cost-optimality-implies-realizability}
    Assume $\varphi$ is realizable.
    If an orchestrator $\orch$ satisfies \Cref{eq:lexicographic-objective}, then it realizes the specification $\varphi$.
\end{theorem}
\begin{proof}
    Since by assumption $\varphi$ is realizable, then there exists an orchestrator $\orch'$ that realizes it.
    By \Cref{thm:optimality-iff-realizability}, we can deduce that the optimal value of $\P^\stocC_\varphi(\orch')$ is $1$.
    Moreover, by assumption and by \Cref{eq:lexicographic-objective}), it follows that $\orch \in \Gamma_{\P^\stocC_\varphi}$, i.e. $\P^\stocC_\varphi(\orch) = 1$, by the arguments above.
    Finally, again by \Cref{thm:optimality-iff-realizability}, we get that $\orch$ realizes $\varphi$.
\end{proof}

%
%
%
%
\noindent
Finally, we formally state the stochastic version of our problem:
\begin{problem}[Stochastic Composition for \LTLf Specifications]
    \label{prob:stoc-comp}
    Given the pair $(\stocC, \varphi)$, where $\varphi$ is an \LTLf task specification over the set of propositions $A$, and $\stocC$ is a community of $n$ stochastic services $\stocC = \{\stocS_1,\dots,\stocS_n\}$, compute, if it exists, an orchestrator that is optimal according to \Cref{eq:lexicographic-objective}.
\end{problem}
\noindent
Interestingly, \Cref{thm:optimality-iff-realizability} and \Cref{thm:cost-optimality-implies-realizability} show that Problem \ref{prob:stoc-comp} is a proper generalization of Problem \ref{prob:nondet-comp}.
In particular, one can reduce the problem of finding an orchestrator in non-stochastic setting
to the problem in stochastic setting by considering arbitrary services' probability distributions for $P_i(\sigma_i,a)$, for any pair $\sigma_i$ and $a$, whose support is compatible with $\delta_i$, and then check whether $\max_\orch \P^\stocC_{\varphi}(\orch) = 1$.

\subsection{Stochastic Services Solution Technique}
Our solution technique is based on finding an optimal policy for a bi-objective lexicographic optimization on a specifically built MDP. In particular, we consider a variant of the framework introduced in \cite{busatto2023bi}: while as the second objective, they considered the expected number of steps to a target, here we consider the expected cost.
Our technique breaks down into the following steps: (1) first, we compute the equivalent \NFA of an \LTLf formula, $\A_\varphi$, and (2) we consider the \DFA $\A_\tk$, as in the non-stochastic setting;
then (3) we compute a \emph{product of $\A_\tk$ with the stochastic services in $\stocC$, obtaining a new MDP}, $\M'$, that we call the ``composition MDP''; (4) we find a policy $\pi$ for $\M'$ that is optimal w.r.t. the bi-objective lexicographic function, as in \cite{busatto2023bi}, and then (5) we derive an orchestrator $\orch$ from $\pi$ that is optimal w.r.t. Equation \ref{eq:lexicographic-objective}.
We now detail each step.
First, observe that \textbf{Step 1 and 2} are the same as the non-stochastic setting. 

\noindent\textbf{Step 3.}
Consider a task specification $\varphi$ and a community of stochastic services $\stocC$. 
Let $\A_\tk = (A\times Q, Q, q_0, F, \delta_\tk)$ be the \DFA associated to the \NFA $\A_\varphi$.
We define the \emph{Composition MDP} $\M = (S', A', P', s'_0)$ as follows:
$S' = Q\times\Sigma_1\times\cdots\times \Sigma_n$; $A' = A\times Q\times \{1\dots n\}$; $s'_0 = (q_0,\sigma_{10}\dots\sigma_{n0})$; $P'(q',\sigma'_1\dots\sigma'_i\dots\sigma'_n|q,\sigma_1\dots\allowbreak\sigma'_i\dots\allowbreak\sigma_n, (a,i)) = P_i(\sigma'_i|\sigma_i, a)$ if $\delta_\tk(q,(a,q')) = q'$.
Moreover, let the composition cost function $C': S'\times A' \to \mathbb{R}^+$ be defined as $C'((q,\sigma_1\dots\sigma_i\dots\sigma_n), \allowbreak (a,q,i)) = C_i(\sigma_i, a)$.
%
%
%
%
%
%

We are interested in computing optimal policies for $\M$, where the optimality is defined as follows.
Consider the target states $T = F\times F_1\times\cdots\times F_n$.
We consider the bi-objective lexicographic optimization over $\M'$, similarly to what has been done in \cite{busatto2023bi}. In particular,
we first consider the probability of reaching a set of target states $T$ from $s\in S'$, following a policy $\pi$ over the MDP $\M'$, denoted with $\prob_{\M'_\pi,s}(\lozenge T)$;
with $\Pi_{\M',s}(\lozenge T)$, we denote the set of policies with the maximum probability of reaching $T$, i.e. $\arg\max_\pi \prob_{\M'_\pi,s}(\lozenge T)$.
Then, we consider the cost of the shortest prefix of $\rho$ that reaches one of the target states in $T$, i.e. $\cost_{T}(\rho) = \sum_{k=0}^{i} C'(s'_k, a'_k)$ if $\rho[i]\in T$ and for all $j < i$, $\rho[j] \not\in T$.
An optimal solution for $\M'$ is a policy $\pi$ that minimizes the conditional expected cost of reaching a target state $\mathbb{E}_{\M'_\pi,s'_0}[\cost_{T}(\rho) | \lozenge T]$ among the policies in $\Pi_{\M,s'_0}(\lozenge T)$, that is, the policies which maximize $\prob_{\M_\pi,s'_0}(\lozenge T)$, i.e.:
\begin{equation}
    \label{eq:mdp-lexicographic-objective}
    \pi\in \Pi_{\M',s_0}(\lozenge T) \text{ and } \pi\in\arg\min_{\pi'} \mathbb{E}_{\rho\sim\M'_\pi,s'_0}\big[\cost_T(\rho) | \lozenge T\big]
\end{equation}
\noindent\textbf{Step 4.}
The solution technique we will use is based on the work \cite{busatto2023bi}, where the authors propose a two-stage technique to find an optimal policy for a bi-objective lexicographic function in the form of \Cref{eq:mdp-lexicographic-objective}. 
First, we compute the set of policies (in the form of a set of optimal actions for each state) that maximize the probability of reaching the target states; however, this set of policies also contains the ``deferral'' policies, i.e. policies that defer the actual reaching of the target states indefinitely, but in such a way that the target can still be reached with maximum probability at any moment.
Then, we consider a ``pruned MDP'' in which \myi only optimal action can be taken, and \myii only states from which the target can be reached are kept. The new MDP is used to find policies that minimize the expected cost of reaching the target. By construction, the optimal policy of the pruned MDP guarantees the target is always reached since any deferral policy will incur an infinite cost.
The difference between our scenario and \cite{busatto2023bi} is that they consider the length of the path, rather than its cost, as the second objective function. Nevertheless, it is easy to see that their approach works if, instead of considering the expected length of successful paths, we consider their expected total costs (i.e. minimizing path length can be seen as minimizing costs with each transition having unitary cost).
Note that the techniques used to find the solutions are standard: the first stage requires solving the maximal reachability
probability problem \cite{DBLP:phd/us/Alfaro97} on the composition MDP with the accepting end
components as the set of states $T$.
The second stage requires solving a stochastic shortest path problem \cite{puterman1994markov} on the pruned MDP.
The two subproblems can be solved efficiently using standard planning algorithms, e.g., Value Iteration or Linear Programming.

\noindent\textbf{Step 5.}
Once an optimal policy is found, we can obtain its equivalent $\orch$ as follows: for any finite prefix of a run $\rho = (q_0,\sigma_{10}\dots\sigma_{n0}),\allowbreak(a_1,q_1,o_1),\dots(a_m,q_m,o_m),\allowbreak(q_m,\sigma_{1m}\dots\sigma_{nm})$,
we set $\orch((\sigma_{10}\dots\sigma_{n0})\allowbreak\dots(\sigma_{1m}\dots\sigma_{nm})) = (a_{m+1},o_{m+1})$, where $\pi(\rho) = (a_{m+1}, q_{m+1}, o_{m+1})$.

\medskip\noindent
Now we aim to establish a relationship between optimal orchestrators according to \Cref{eq:lexicographic-objective}, and optimal policies for $\M'$ according to \Cref{eq:mdp-lexicographic-objective}.
Given a run $\rho = (q_0,\sigma_{10}\dots\sigma_{n0}),(a_1,q_1,o_1)\dots$, we define the trace $t = \stoc\tau_{\varphi,\stocC}(\rho) = (\sigma_{10}\dots\sigma_{n0}),(a_1,o_1),\dots$. 
%
The following lemma shows that once fixed a policy $\pi$ over $\M'$, there is a one-to-one correspondence between the paths on $\M'$ following $\pi$ and the executions of the equivalent orchestrator of $\pi$, $\orch$.

\begin{lemma}
    \label{lemma:rho-iff-h}
    Let $\pi$ be a policy for $\M'$ and let $\orch$ be its equivalent orchestrator.
    Moreover, let $\rho\in \Paths_{{\M'}_\pi}^\omega$ and $t$ be a trace such that $t=\stoc\tau_{\varphi,\stocC}(\rho)$.
    Then, $\rho\in\Paths_{\M'_\pi}^\omega(\langle s'_0 \rangle)$ iff $t\in \T_{\orch,\stocC}(\langle (\sigma_{10}\dots\sigma_{n0})\rangle)$.
\end{lemma}
\begin{proof}
\noindent
($\Rightarrow$) Let $\rho\in \Paths_{\M'_{\pi}}^\omega(\langle s'_0 \rangle)$.
Consider the history $h=\tau_{\varphi,\stocC}(\rho)$.
We prove the claim by induction on the position of the path/history.

\noindent
\emph{Base case}: we have the claim holds for position $0$ because $\rho[0] = (q_0,\sigma_{10}\dots\sigma_{n0})$, and $h[0] = \sigma_{10}\dots\sigma_{n0}$. Therefore, $\langle h[0]\rangle$ satisfies the conditions of the definitions of history and execution of $\orch$ iff $s'_0\in S'$.

\noindent
\emph{Inductive case}: assume the claim holds up to position $k\ge 0$.
Consider the $(k+1)$-th action according to $\pi$, i.e. $\pi(\rho[0\dv k]) = (a_{k+1}, q_{k+1}, o_{k+1})$,
and its successor state $\rho[k+1] = (q_{k+1},\sigma_{1,k+1},\dots,
\sigma_{n,k+1})$.
By construction of $\M'$, we have that 
\myi $\sigma_{i,k+1}\in \Sigma_i$ for all services, 
\myii $o_{k+1}\in\{1\dots n\}$, 
\myiii $a_{k+1}\in A$, and 
\myiv $\sigma_{k+1}\in\supp(P_{o_{k+1}}(\sigma_{o_{k+1},k},a_{k+1}))$.
Moreover, by construction of $\orch$, $(a_{k+1}, o_{k+1})=\allowbreak\orch(\allowbreak\states(h[0\dv k]))$, hence $h'=h[0\dv k]\cdot (a_{k+1},o_{k+1}),(\sigma_{1,k+1}\dots\sigma_{n,k+1})$ is a proper (finite) execution.
The same arguments can be applied in the other direction. By induction the claim also holds for any arbitrary position, and therefore $\rho\in\Paths_{\M'_\pi}^\omega(\langle s'_0 \rangle)$ iff $t\in \T_{\orch,\stocC}(\langle (\sigma_{10}\dots\sigma_{n0})\rangle)$.
\end{proof}

\begin{lemma}
    \label{lemma:composition-mdp-and-community-same-prob}
    Let $\pi$ a policy on $\M'$, $\orch$ be its equivalent orchestrator,
    $\rho = s'_0a_1\dots s'_m\in \Paths_{\M'_\pi}(s'_0)$ 
    be a finite path on $\M'$, and $\stoch = \stoc\tau_{\varphi,\stocC}(\rho)$ be its associated history.
    Then, $\prob_{\M_\pi,s'_0}(\Paths_{\M_\pi}^\omega(\rho)) = \prob_{\orch,\stocC}(\T_{\orch,\stocC}(h))$.  
\end{lemma}
\begin{proof}
    \begin{align}
        \prob_{\M'_\pi,s'_0}(\Paths_{\varphi,\stocC}^\omega(\rho)) &= \prod\limits_{k=1}^{m} P'(s'_{k} \mid s'_{k-1},(a_k,q_k,o_k)) \label{lemma:composition-mdp-and-community-same-prob:prob-of-cylinder} \\
        & = \prod\limits_{k=1}^m  P_{o_k}(\sigma_{o_k,k} \mid \sigma_{o_k,k-1}, (a_k,o_k)) \label{lemma:composition-mdp-and-community-same-prob:by-construction-of-p} \\
        & = \prob_{\orch,\stocC}(\T_{\orch,\stocC}(h)) \label{lemma:composition-mdp-and-community-same-prob:by-definition-of-prob-comm}
    \end{align}
where step \ref{lemma:composition-mdp-and-community-same-prob:prob-of-cylinder} is by definition of the probability of a cylinder set, 
step \ref{lemma:composition-mdp-and-community-same-prob:by-construction-of-p} by definition of $P'$ in $\M'$,
and
step \ref{lemma:composition-mdp-and-community-same-prob:by-definition-of-prob-comm} by \Cref{eq:prob-sat-goal}.
\end{proof}

\begin{lemma}
    \label{lemma:end-in-t-iff-successful}
    Let $\rho = s_0a_1\dots s_m\in \Paths_{\M'_\pi}$ 
    be a finite path on $\M'$, and let $h = \stoc\tau_{\varphi,\stocC}(\rho)$ be its associated history.
    Then, $s_m\in T$ iff $\successful(h)$.
\end{lemma}
\begin{proof}
    By induction on the length of the run/history.

    \noindent
    \emph{Base case}: $\rho_0 = \langle s'_0 \rangle = \langle(q_0,\sigma_{10}\dots\sigma_{n0})\rangle$. 
    Let $h_0 = \stoc\tau_{\varphi,\stocC}(\rho_0) = \langle(\sigma_{10}\dots\sigma_{n0})\rangle$.
    We have that $\rho_0[0] \in T$ iff (i.a) $q_0\in F$ and (i.b) $\sigma_{i0}\in F_i$ for all $1\le i\le n$.
    On the other hand, $h$ is successful iff (ii.a) $\actions(h_0) = \epsilon\models\varphi$ and (ii.b) $\sigma_{i0}\in F_i$.
    The claim holds because (i.b) is precisely (i.b), and (i.a) holds iff (ii.a) holds by the correctness of the construction of $\A_\tk$.

    \noindent
    \emph{Inductive case}: assume the claim holds for $\rho_{k-1} = (q_0,\sigma_{10}\dots\sigma_{n0}),(a_1,q_1,o_1),\dots,\allowbreak(a_{k-1},q_{k-1},o_{k-1}),(q_{k-1},\sigma_{1,k-1}\dots\sigma_{n,k-1})$
    and $h_{k-1} = \stoc\tau_{\varphi,\stocC}(\rho)$.
    %
    %
    Let $a'_k = (a_k,q_k,o_k)$ be any valid next action taken from $s_{k-1}$, and let $s_k = (q_k,\sigma_{1k}\dots\sigma_{nk})\in\supp(P_{o_k}(s_{k-1},a_k))$ the next possible state.
    Consider the sequence $r=q_0\dots q_k$. By construction of $\M'$, and correctness of $\A^d$, we have that $r$ is a run over $\A^d$, and that $q_k \in F$ iff $a_1\dots a_k\models \varphi$.
    By definition of $h_k = \stoc\tau_{\varphi,\stocC}(\rho_k)$, we also have that $\actions(h_k) = a_1\dots a_k$.
    Finally, we have that $s_k\in F$ iff \myi $q_k\in F$ and \myii for all $i$ $\sigma_{ik}\in F_i$ by construction of $\M'$; \myi holds iff \myiii $\actions(h_k)\models \varphi$ by the arguments above; finally, \myii and \myiii hold iff $\successful(h)$.
\end{proof}
\noindent
Let $\Paths_{T,\M'_\pi}(s'_0)$ be the set of finite paths following $\pi$ on $\M'$ such that they start with $s'_0$ and enter in a state in $T$ only at the end of the path and for the first time, i.e. $\Paths_{T,\M'_\pi}(s'_0) = ((S' \setminus T ) \times A)^*T \cap \Paths_{\M'_\pi}(s'_0)$.

\begin{lemma}
    \label{lemma:shortest-successful-stories}
    $\rho\in\Paths_{T,\M'_\pi}(s'_0)$ iff $\stoc\tau_{\orch,\stocC}(\rho)\in\H^\varphi_{\orch,\stocC}$
\end{lemma}
\begin{proof}
    By \Cref{lemma:end-in-t-iff-successful}, $\rho\in \Paths_{T,\M'_\pi}$ iff $h=\stoc\tau_{\orch,\stocC}(\rho)$ is successful.
    Moreover, by \Cref{lemma:rho-iff-h}, $\rho\in\Paths_{T,\M'_\pi}\subseteq\Paths_{\M'_\pi}$ iff $h$ is an execution of $\orch$.
    Furthermore, by assumption, any finite prefix $\rho'$, say of length $m$, of $\rho$, is such that $\rho'[m] \not\in T$. 
    Then, again by \Cref{lemma:end-in-t-iff-successful}, this holds iff $h'=\stoc\tau_{\varphi,\stocC}(\rho')$ is not successful, meaning that does not exist a prefix $h'\in\prefixes(h)$ with $h'\neq h$ such that $h'$ is successful.
    But this is precisely the membership condition for $\H^\varphi_{\orch,\stocC}$
\end{proof}
\noindent
Intuitively, \Cref{lemma:rho-iff-h} shows there is a one-to-one correspondence (modulo choices of $q_0\dots q_m$ in $\rho$) between paths $\rho$ in $\M'$ with $\pi$ and traces $t\in\T_{\orch,\C}$;
\Cref{lemma:composition-mdp-and-community-same-prob} shows that the probabilities of finite paths and histories are the same;
\Cref{lemma:end-in-t-iff-successful} shows that paths that end with states in $T$ correspond to successful histories;
and
\Cref{lemma:shortest-successful-stories} shows a correspondence between paths in $\Paths_{T,\M'_\pi}$ and $\H_{\orch,\stocC}^\varphi$.
This result shows the correctness of our technique:

\begin{theorem}
    \label{thm:pi-optimal-implies-orch-optimal}
    Let $(\stocC, \varphi)$ be an instance of Problem 2, and let $\M'$ be the composition MDP for $\stocC$ and $\varphi$.
    If $\pi$ is optimal (w.r.t. \Cref{eq:mdp-lexicographic-objective}) iff its equivalent orchestrator $\orch$ is optimal (w.r.t. \Cref{eq:lexicographic-objective}).
\end{theorem}
\begin{proof}
    First, we show that $\pi = \arg\max_{\pi'} \prob_{\M_\pi',s'_0}(\lozenge T)$ iff $\orch = \arg\max_{\orch'} \P^\stocC_{\varphi}(\orch')$.
    For any pair $\pi$ and its equivalent $\orch$, we have:
    \begin{align}
        \prob_{\pi,s'_0}(\lozenge T) &= \sum\limits_{\rho_T\in \Paths_{T,\pi}(s'_0)} \prob_{\M'_{\pi'},s'_0}(\Paths_{\M'_\pi,s'_0}^\omega(\rho_T)) \label{thm:pi-optimal-implies-orch-optimal:def-reachability}\\
        & = \sum\limits_{\stoch\in\H^\varphi_{\orch,\stocC}} \prob_{\orch,\stocC}(\T_{\orch,\stocC}(\stoch)) \label{thm:pi-optimal-implies-orch-optimal:by-lemma}\\
        & = \prob_{\orch,\stocC}\bigg(\bigcup_{h\in\H^\varphi_{\orch,\stocC}} \T_{\orch,\stocC}(h) \bigg) \label{thm:pi-optimal-implies-orch-optimal:disjointness}\\
        & = \P^\stocC_\varphi(\orch) \label{thm:pi-optimal-implies-orch-optimal:by-definition}
   \end{align}
where step \ref{thm:pi-optimal-implies-orch-optimal:def-reachability} is by definition of probabilistic reachability,
step \ref{thm:pi-optimal-implies-orch-optimal:by-lemma} is by \Cref{lemma:composition-mdp-and-community-same-prob} and \Cref{lemma:shortest-successful-stories},
step \ref{thm:pi-optimal-implies-orch-optimal:disjointness} is by 
disjointness of all $\T_{\orch,\stocC}(h)$ for $h\in\H^\varphi_{\orch,\stocC}$,
and step \ref{thm:pi-optimal-implies-orch-optimal:by-definition} is by \Cref{eq:prob-sat-goal}.
From this, we obtain that $\pi^* = \arg\max_{\pi'} \prob_{\M_{\pi'},s'_0}(\lozenge T)$ iff $\orch^* = \arg\max_{\orch'} \P^\stocC_{\varphi}(\orch')$.

It remains to prove that $\pi$ is cost-optimal iff $\orch$ is cost-optimal. We have:

\begin{align}
    \mathbb{E}_{\rho\sim\M'_\pi,s'_0}&[\cost_T(\rho) \mid \lozenge T] =\nonumber\\ &= \sum\limits_{\rho_T\in \Paths_{T,\pi}(s'_0)} \prob_{\M'_{\pi'},s'_0}(\Paths_{\M'_\pi,s'_0}^\omega(\rho_T)) \cdot \sum\limits_{k=0}^{|\rho_T|} C'(s'_k,a'_{k+1}) \label{thm:pi-optimal-implies-orch-optimal:mdp-cost-expectation-definition}\\
    &= \sum\limits_{\stoch\in\H^\varphi_{\orch,\stocC}} \prob_{
    \orch,\stocC}(\T_{\orch,\stocC}(\stoch)) \cdot \sum\limits_{k=0}^{|h|} C_{o_{k+1}}(\sigma_{o_k,k},a_{k+1}) \label{thm:pi-optimal-implies-orch-optimal:mdp-construction} \\
    &= \mathbb{E}_{h\sim \prob_{\orch,\stocC}}\Bigg[ \sum\limits_{k=1}^{|h|} C_{o_k}(\sigma_{o_k,k-1},a_{k}) \bigg\vert\  \successful(h) \Bigg] \label{thm:pi-optimal-implies-orch-optimal:expectation-on-histories}\\
    & = \J^\stocC_\varphi(\orch) \label{thm:pi-optimal-implies-orch-optimal:def-orch-cost-function}
\end{align}
where
step \ref{thm:pi-optimal-implies-orch-optimal:mdp-cost-expectation-definition} by definition of total expected cost conditioned on reaching of target states $T$,
step \ref{thm:pi-optimal-implies-orch-optimal:mdp-construction}
by construction of $\M'$ and by \Cref{lemma:composition-mdp-and-community-same-prob} and \Cref{lemma:shortest-successful-stories},
step \ref{thm:pi-optimal-implies-orch-optimal:expectation-on-histories} by definition of total expected cost on successful executions of $\orch$, and
step \ref{thm:pi-optimal-implies-orch-optimal:def-orch-cost-function} by \Cref{eq:min-exp-cost}.
Therefore, if $\pi\in \arg\min_{\pi'} \mathbb{E}_{\rho\sim\M'_\pi,s'_0}[\cost_T(\rho) \mid \lozenge T]$ then $\orch\in \arg\min_{\orch'} \J^\stocC_\varphi(\orch')$.
Combining both results, we get the thesis.
\end{proof}
\noindent
\textbf{Computational cost.} \Cref{thm:pi-optimal-implies-orch-optimal} guarantees that we can reduce Problem \ref{prob:stoc-comp} to the problem of finding an optimal policy for the lexicographic bi-objective optimization problem (\Cref{eq:mdp-lexicographic-objective}) over a composition MDP $\M'$.
As explained above, the two-stage technique requires solving a planning problem over MDPs.
Since it is known that both steps require polynomial time complexity in the number of states and actions of the MDP \cite{puterman1994markov} and that our Composition MDP has a state space that is a single-exponential in the size of the task specification, we get this result:
%
%
\begin{theorem}
Problem \ref{prob:stoc-comp} can be solved
    in at most exponential time in the size of the formula, 
    in at most exponential time in the number of services, 
    and in polynomial time in the size of the services.
\end{theorem}
\noindent
Observe that, differently from the classical setting of \LTL/\LTLf synthesis on probabilistic systems \cite{DBLP:journals/jacm/CourcoubetisY95,LTLfSynthesisonProbabilisticSystems}, and analogously to our solution method for the non-stochastic case,
we do not have unobservable ``adversarial'' nondeterminism in the composition MDP; hence, we do not need to determinize the \NFA of the task specification, thereby saving an exponential blow-up in time complexity.

\section{Example}
This section describes an example of using our stochastic framework, inspired by the ``garden bots system'' example \cite{yadav2011decision}.
The task is to $\mathit{clean}$ the garden the garden by picking fallen leaves and removing dirt, $\mathit{water}$ the plants, and $\mathit{pluck}$ the ripe fruits and flowers. The action $\mathit{clean}$ must be performed at least once, followed by $\mathit{water}$ and $\mathit{pluck}$ in any order.
In \declare\ \LTLf, the goal can be expressed as $\varphi = \mathit{clean} \wedge (\mathit{clean}\lUntil((\mathit{water}\land\Next\mathit{pluck}\lor(\mathit{pluck}\land\Next\mathit{water})))$.
We assume there are three available garden bots, each with different capabilities and rewards. In \Cref{fig:example} the three services specifications and the \DFA of the \LTLf goal are shown. Transitions labels are of the form $\langle \mathit{action}, \mathit{prob}, \mathit{reward}\rangle$.
We are interested in a composition of the bots to satisfy the goal $\varphi$.
In the nondeterministic case, bot 1 will be used to perform $\mathit{clean}$, while either bot 2 or 3 will be used for $\mathit{pluck}$, followed by $\mathit{water}$ (or the other way around).
Moreover, bot 1 might need to be emptied, which cannot be known beforehand due to nondeterminism, and must be handled by the solution.
In the stochastic case, we also have to take into account rewards/costs.
Despite being less costly than other bots, the optimal orchestrator cannot use $\mathit{clean}$ of bot 1 indefinitely since the lexicographic constraint obliges to pursue the satisfaction of the goal.
The stochasticity of bot 1 must be taken into account in order to have the right total cost estimate.
Moreover, we should prefer bot 3 for $\mathit{pluck}$ because the $\mathit{empty}$ of bot 2 costs too much ($-5$).





\begin{figure}[t]
  \centering
  \begin{subfigure}{0.3\textwidth}
\begin{tikzpicture}
    \node[state,initial,initial text=,minimum size=0.5cm] (g0) { \( g_0 \) };
    \node[state, below=1.25cm of g0,minimum size=0.5cm] (g1) { \( g_1 \) };
    \node[state, right=0.25cm of g0,minimum size=0.5cm] (g2) { \( g_2 \) };
    \node[state, right=0.5cm of g1,minimum size=0.5cm] (g3) { \( g_3 \) };
    \node[state,accepting, below right=0.5cm and 0.25cm of g2,minimum size=0.5cm] (g4) { \( g_4 \) };

    \path[->] (g0) edge[] node[left,rotate=90,xshift=0.32cm,yshift=0.2cm] {{\footnotesize $\mathit{clean}$}} (g1);
    \path[->] (g1) edge[loop left] node[above,xshift=0cm,yshift=-0.42cm] {{\footnotesize $\mathit{clean}$}} (g1);
    \path[->] (g1) edge[] node[left,rotate=60,xshift=0.4cm, yshift=0.14cm] {{\footnotesize $\mathit{water}$}} (g2);
    \draw[->] (g1) edge[bend right] node[below,xshift=-0.05,yshift=-0.05cm] {{\footnotesize $\mathit{pluck}$}} (g3);

    \draw[->] (g3) edge[bend right] node[below,rotate=45,xshift=-0.1cm,yshift=-0.175cm] {{\footnotesize $\mathit{water}$}} (g4);
    \draw[->] (g2) edge[bend left] node[above,rotate=-45,yshift=0.1cm] {{\footnotesize $\mathit{pluck}$}} (g4);
    \draw[->] (g4) edge[loop below,rotate=60] node[below right] {{\footnotesize $\top$}} (g4);
\end{tikzpicture}
\label{fig:example-1:dfa}
\caption{The \DFA of $\varphi$.}
\end{subfigure}
  \begin{subfigure}{0.21\textwidth}
      \begin{tikzpicture}[shorten >=1pt,node distance=2.5cm,on grid,auto,initial text=,every node/.style={inner sep=0,outer sep=0}]
        \node[state,accepting,initial,initial text=,minimum size=0.75cm] (a0) { $a_0$ };
        \node[state, below=of a0,minimum size=0.75cm] (a1) { \( a_1 \) };
        \path[->] (a0) edge[bend right] node[left,rotate=-90,xshift=0.9cm,yshift=-0.15cm] {{\scriptsize $\langle \mathit{clean}, 0.2,-0.1 \rangle$}} (a1);
        \path[->] (a0) edge[loop right] node[above,xshift=0.1cm,yshift=0.42cm] {{\scriptsize $\langle \mathit{clean},0.8,-0.1 \rangle$}} (a0);
        \path[->] (a1) edge[bend right] node[right,rotate=90,xshift=-0.9cm,yshift=-0.15cm] {{\scriptsize $\langle\mathit{empty},1,-0.1\rangle$}} (a0);
    \end{tikzpicture}
    \caption{Bot 1.}
    \label{fig:example-1:bot-1}
\end{subfigure}
\begin{subfigure}{0.21\textwidth}
\begin{tikzpicture}[shorten >=1pt,node distance=2.5cm,on grid,auto,initial text=,every node/.style={inner sep=0,outer sep=0}]
    \node[state,accepting,initial,initial text=,minimum size=0.75cm] (b0) { \( b_0 \) };
    \node[state, below=of b0,minimum size=0.75cm] (b1) { \( b_1 \) };
    \path[->] (b0) edge[bend right] node[left,rotate=-90,xshift=0.65cm,yshift=-0.1cm] {{\scriptsize $\langle\mathit{pluck},1,-1\rangle$}} (b1);
    \path[->] (b0) edge[loop right] node[above,xshift=0.0cm,yshift=0.42cm] {{\scriptsize $\langle\mathit{water},1,-1\rangle$}} (b0);
    \path[->] (b1) edge[bend right] node[right,rotate=90,xshift=-0.7cm,yshift=-0.1cm] {{\scriptsize $\langle\mathit{empty},1,-5\rangle$}} (b0);
    \path[->] (b1) edge[loop right] node[xshift=-0.4cm,yshift=0.3cm] {{\scriptsize $\langle\mathit{water},1,-1\rangle$}} (b1);
\end{tikzpicture}
\caption{Bot 2.}
\label{fig:example-1:bot-2}
\end{subfigure}
\begin{subfigure}{0.21\textwidth}
\begin{tikzpicture}[shorten >=1pt,node distance=2.5cm,on grid,auto,initial text=,every node/.style={inner sep=0,outer sep=0}]
    \node[state, accepting,initial,initial text=,minimum size=0.75cm] (c0) { \( c_0 \) };
    \node[state, below=of c0,minimum size=0.75cm] (c1) { \( c_1 \) };
    
    \path[->] (c0) edge[bend right] node[rotate=-90,xshift=-0.9cm,yshift=-0.1cm] {{\scriptsize $\langle\mathit{pluck},1,-1\rangle$}} (c1);
    \path[->] (c1) edge[bend right] node[rotate=90,xshift=0.9cm,yshift=-0.1cm] {{\scriptsize $\langle\mathit{empty},1,-1\rangle$}} (c0);
\end{tikzpicture}
\caption{Bot 3.}
\end{subfigure}
  \caption{The garden bot systems and the \DFA of the \LTLf goal.}
  \label{fig:example}
\end{figure}
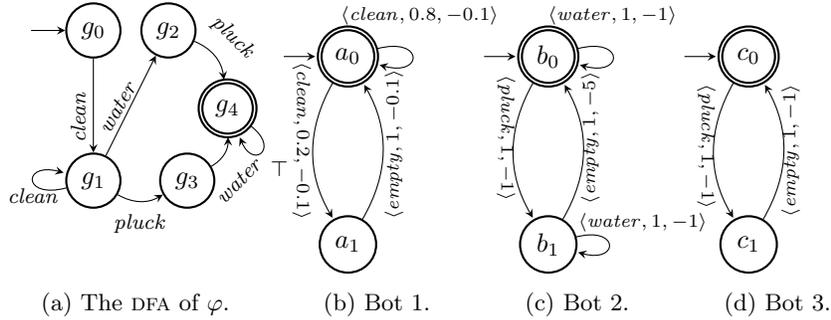

\section{Related Works}
\label{sec:related}
We discuss related works to our framework in two parts: one for the nondeterministic case and the other for the stochastic case.

\noindent
\textbf{Nondeterministic service composition.}
In the nondeterministic case, our framework is closely related to the problem of fully observable non-deterministic planning for \LTLf goals \cite{DeGiacomoR18,DBLP:conf/aaai/CamachoTMBM17}.
However, there are two main differences. \myi The first difference is that, in their case, the \LTLf formula is interpreted over trajectories of states (where a state is a set of fluents). In contrast, in our case, the formula is evaluated over the sequences of actions chosen by the orchestrator. This is because we want our framework to be \emph{domain-agnostic} from the point of view of the task specification; in other words, the specification should not deal with the actual working of the services, hence achieving a better modularization and separation of concerns within the composition framework (note, this is also the difference with the classical \LTLf synthesis problem \cite{DBLP:conf/ijcai/GiacomoV15}).
\myii In their case, the nondeterminism comes from the planning domain, while in our case, it comes from the nondeterministic services' behaviour.
Furthermore, the solution technique proposed in \cite{DBLP:conf/aaai/CamachoTMBM17} works by providing a particular PDDL encoding extending the original domain with auxiliary action to progress the formula's DFA state, while in \cite{DeGiacomoR18}, the authors rely on adversarial reachability on a classical cartesian product between the domain's DFA and the formula's DFA; instead, our solution cannot rely on a first-order specification as PDDL, and the automata-based construction is more complex since it has to handle the asynchronous progression of the services' states.
Planning methods able to cope with non-determinism have been widely applied in the literature of web service composition, as reviewed in \cite{DBLP:journals/aires/MarkouR16}, although most of the methods there do not support stateful services.
In \cite{DBLP:journals/ws/SirinPWHN04} Hierarchical Task Network (HTN) planning is used for service composition.
HTN planning is based on the notion of composite tasks that can be refined to atomic tasks using predefined methods. While based on high-level specification of services, their approach does not support nondeterministic services.
The work \cite{DBLP:conf/iscc/AlvesMFD16} allows the modeling of nondeterministic behaviours but not of stateful services nor high-level temporal goal specification.
Authors in~\cite{pistore2005automated} describe services as atomic actions where only I/O behavior is modeled, and the ontology is constituted by propositions and actions; hence services are not stateful as in our case. Moreover, in our case, we can adopt the \declare paradigm by relying on \LTLf and automata-based reasoning.
There is a rich literature on the Roman Model for service composition \cite{berardi2003automatic}, and its extension to nondeterministic services' behaviours \cite{DBLP:conf/icsoc/BerardiCGM05,DBLP:conf/aaai/GiacomoFPS10,DBLP:conf/atal/GiacomoPS10,DBLP:conf/aips/RamirezYS13}.
However, the crucial difference is that, while in previous works, the orchestrator must be able to orchestrate the target service \emph{for every} possible request (safety goal), in our framework, the objective is to \emph{look for the existence} (reachability goal) of a strategy such that the achievement of the goal is guaranteed.
%





\noindent
\textbf{Stochastic service composition.}
In the context of the Roman model, \cite{yadav2011decision,brafman2017service} propose a solution for the service composition in stochastic settings, defining the non-deterministic behaviour of the target service. 
In their model, an optimal solution can be found by solving a Markov Decision Problem derived from the services and requirement specifications. We note that the authors of \cite{brafman2017service}, although able to orchestrate the services, did not capture the non-deterministic behaviours of the available services, and both works do not consider the rewards (or costs) of using a particular service.
Another important aspect that we
introduce is the use of a lexicographic bi-objective optimization to manage the orchestration, which
allows us to define a strict preference of the task rewards over the services' ones. The closest works on this aspect
are \cite{DBLP:conf/bpm/GiacomoFLMS21,DBLP:conf/ijcai/GiacomoFLMS22,DEGIACOMO2023103916}, but the crucial difference is that they did not consider objective priorities as strict, but a composition MOMDP via a linear combination of weights.
One of the earliest works that combined stochastic planning models with service composition is \cite{DBLP:conf/waim/GaoYTZ05}.
There are works based on Markov-HTN Planning
\cite{DBLP:conf/icws/ChenXR09}, multi-objective optimization \cite{DBLP:conf/icsoc/MoustafaZ13,DBLP:journals/tcc/ChenHLS19},
and lexicographic optimization \cite{sadeghiram2020user}, helpful to model the stochastic behaviour as well as complex QoS preferences. However, in all cases, either there are no stateful services or no high-level declarative specification of the desired solution.

\section{Discussion and Conclusion}

In this paper, we have studied an advanced form of task-oriented compositions of both nondeterministic and stochastic services. 
Recently, service composition has become very relevant in smart manufacturing \cite{DBLP:journals/ai/GiacomoFLPS22,DBLP:conf/aaai/GiacomoVFAL18}. 
In particular, the use of service composition has been advocated in a Digital Twins (DT) scenario, in \cite{DEGIACOMO2023103916}, where the composition of a target service by means of a community of stochastic services is considered. 
Their framework, inspired by \cite{brafman2017service} is similar to our stochastic framework. Among similarities, their available services also exhibit stochastic behaviour and the utilization utility is modelled using a reward function rather than a cost function, as in our case.
%
However, in their case, the goal is to maximize the total expected sum of rewards coming from the dispatching of the target service's requests to the available services. 
Note that the probabilities and the cost of this framework are updated over time, as could happen in a realistic industrial scenario due to continuous use of the services.
%
Despite this, their solution method does not guarantee the target's rewards objective has higher priority than the services' rewards.
Instead, thanks to the lexicographic bi-objective, we are able to guarantee that the utilization cost is never preferred if it conflicts with a higher probability of satisfaction with the task specification.
%
We can find an example of this utilization in \cite{de2023aida} where the authors, through a software prototype, show a manufacturing process concerning an electric motor formalized as Lexicographic Markov Decision Process (LMDP). They define the vector of reward functions formed by two objectives: cost and product quality. The optimal solution they calculate aims to minimize the cost and maximize the quality.
Crucially, both in the nondeterministic and stochastic framework, in our setting, the orchestrator has control over the choice of actions to fulfil the task (as in AI Planning), while in their setting, the orchestrator is more similar to one in the original Roman model \cite{berardi2003automatic}.
%
We believe it would be quite fruitful to apply our stochastic framework to their context, and we will investigate it in the future.

It is important to notice that in this work, the task specification has been expressed in \LTLf only using actions which are under the control of the orchestrator. As a result, one can obtain a \DFA for expressing the action choices the orchestrator has to fulfil the task directly (in linear time) from the \NFA corresponding to the formula. One possible extension is to consider feedback from the services regarding a common set of observations. In principle, our techniques can be extended to that case as well. However, this time we cannot assign the nondeterministic choice to the controller (the orchestrator) as we have done here, and this would require a blow-up that can be worst-case exponential in obtaining the analogous of the \DFA above, see, e.g., \cite{DBLP:conf/ijcai/GiacomoV15,BrafmanGP18}.

\section*{Acknowledgements}
This work has been partially supported by the EU H2020 project AIPlan4EU (No. 101016442), the ERC-ADG WhiteMech (No. 834228), the EU ICT-48 2020 project TAILOR (No. 952215), the PRIN project RIPER (No. 20203FFYLK), and the PNRR MUR project FAIR (No. PE0000013).

\printbibliography

\end{document}